\documentclass{article}

\usepackage{PRIMEarxiv}
\usepackage{subcaption}
\usepackage{amsmath,amssymb,amsthm}
\usepackage[utf8]{inputenc} 
\usepackage[T1]{fontenc}    
\usepackage{hyperref}       
\usepackage{url}            
\usepackage{booktabs}       
\usepackage{amsfonts}       
\usepackage{nicefrac}       
\usepackage{microtype}      
\usepackage{lipsum}
\usepackage{fancyhdr}       
\usepackage{graphicx}       
\graphicspath{{media/}}     
\newtheorem{assumption}{Assumption}
\newtheorem{remark}{Remark}
\newtheorem{theorem}{Theorem}
\newtheorem{lemma}{Lemma}
\title{Observer-Based Fixed-Time Nested Sliding-Mode Control for Tip-Position Regulation of a Single-Link Flexible Manipulator
\thanks{Atul Sharma, Chayan Kumar Paul, and S. Janardhanan are with the Department of Electrical Engineering, Indian Institute of Technology, Delhi, 110016 India (e-mail: atul.sharma@ee.iitd.ac.in).}}

\author{
  Atul~Sharma,
    Chayan~Kumar~Paul,
      S.~Janardhanan 
}

\begin{document}
\maketitle

\begin{abstract}
This paper presents a novel position control strategy for a single-link flexible manipulator, tailored for applications where precise position must be achieved within strict time constraints. To accomplish this objective, firstly, a nested non-singular terminal sliding mode controller is designed for the system, enabling precise and robust control. Furthermore, a fixed-time sliding mode observer is designed to estimate unmeasured system states accurately in a fixed time, thereby enabling closed-loop control implementation. A stability analysis is presented to guarantee the robustness and efficacy of the proposed composite control algorithm. The effectiveness of the proposed fixed-time controller is demonstrated through numerical simulation on accuracy, stability, and convergence speed. The proposed controller's performance is also compared with that of other state-of-the-art control schemes. The proposed controller is further validated through experiments conducted on a real hardware setup.
\end{abstract}


%


\section{Introduction}

Robotics deployment has produced notable breakthroughs in the fields of exploration and manufacturing. Rigid link manipulators (RLMs) have traditionally been employed in a wide range of applications, including logistics operations \cite{yuan2007research}, infrastructure construction \cite{warszawski1985robotics}, and other multiple fields \cite{li2024advances1}. However, their considerable size and substantial power requirements often restrict their operational efficiency. Because of their lightweight nature, flexible link manipulators (FLMs) provide an alternative. However, this comes at the cost of flexibility; improperly handling it can cause undesired vibrations and is difficult to control. This research aims to design a control algorithm that effectively addresses these challenges and optimises FLM's potential for practical applications.

Many researchers have studied the problem of tip oscillations in FLMs using various control techniques. For effective control design, FLMs must be accurately modelled \cite{korayem2012mathematical, de2001rest, oguamanam2006dynamic, dixit2014cable}. The dynamic model of FLMs is described by partial differential equations (PDEs), but analytical solutions of these equations are not always possible. Li et al. \cite{li2024advances1} offer a recent review that provides a thorough summary of different modelling approaches for FLMs. 

Accurate modelling serves as a foundation for effective control design; once the dynamics are well captured, appropriate control schemes can be designed to mitigate tip oscillations. Several control strategies have been implemented to address the tip oscillations in FLMs. Li et al. \cite{li2024advances2} present an extensive review of different control methodologies for FLMs. Among various flexible manipulator configurations, the Single-Link Flexible Manipulator (SLFM) has been the principal testbed for control research, due to its simplicity and rich nonlinear dynamics.
SLFMs are nonlinear, underactuated, and distributed-parameter systems characterised by strong coupling between rigid and flexible body dynamics. Their tip-position regulation is challenging because the flexible dynamics introduce non-minimum-phase behaviour and unmeasured bending modes.
Early work on boundary-control formulations, based on the Euler–Bernoulli beam model, provides rigorous asymptotic stabilisation by directly acting on boundary variables such as torque or bending moment. Jiang et al. \cite{jiang2015boundary} proposed a boundary-disturbance observer to reject unknown inputs at the clamped end, guaranteeing uniform boundedness of the distributed states. Although such designs achieve theoretical completeness, the solution relies on collocated sensing, exact knowledge of the PDE, and boundary actuation conditions that are rarely met in practical experimental setups. While boundary-control methods provide strong theoretical guarantees, their practical implementation is often constrained by sensing and actuation limitations. In particular, it is not feasible to continuously measure system variables over both space and time, necessitating the use of limited or boundary-based sensing. Moreover, practical actuators cannot always realise the ideal boundary inputs assumed in theoretical formulations. The presence of unmodeled dynamics, measurement noise, and time delays further complicates implementation, thereby limiting the effectiveness of these control strategies in flexible link manipulator systems. To overcome these challenges, researchers have turned to finite-dimensional approximations that enable the design of physically implementable controllers.

Finite-dimensional approximations, such as assumed modes or the finite element method (FEM), enable the design of closed-loop, implementable controllers. Choi and Nho \cite{choi1995sliding} demonstrated vibration suppression using classical sliding mode control (SMC) for a single-link arm, showing strong robustness but only asymptotic convergence. Although classical SMC achieves robustness, its convergence is only asymptotic, motivating the development of terminal-type sliding manifolds designed to ensure faster, finite-time convergence. To improve the convergence rate, terminal sliding-mode control (TSMC) was proposed, achieving finite-time stability but suffering from singularity at the origin. Feng and Yu \cite{feng2002non} addressed this issue using the non-singular TSMC (NTSMC), which is widely adopted for robotic manipulators. However, the settling time of TSMC/NTSMC depends on the initial conditions, precluding guaranteed bounds on the response.
The fixed-time stability framework of Polyakov \cite{polyakov2011nonlinear} provides convergence in a uniformly bounded time independent of initial conditions. It has been integrated with SMC for electromechanical systems (e.g., fixed-time SMC for PMSMs \cite{xu2025adaptive}) but not yet specialised for flexible manipulators. However, achieving high-performance control also requires accurate state estimation, as many flexible-link states, such as bending deflection and its derivatives, are difficult to measure directly. This challenge has led to the development of observer-based control frameworks.

Accurate state estimation is important because bending deflection and its derivatives are often unmeasurable. Chalhoub and Azizi \cite{chalhoub2006design} proposed a nonlinear observer with SMC for an SLFM. Fractional-order SMC designs (Nejad et al., \cite{hamzeh2020precise}) improved precision and chattering reduction but remained asymptotic. More recently, Sharma and Janardhanan \cite{sharma2024position} employed a functional observer within an SMC loop to enhance robustness, yet the convergence time remained dependent on the initial condition.

While these observer designs improve robustness for flexible manipulators, their convergence times remain dependent on initial conditions. Parallel developments in fixed-time observers for ODE systems, e.g., Mechali et al. \cite{mechali2022fixed} (UAV attitude), Hosseinabadi et al. \cite{alinaghi2023fixed} (nonlinear double integrator), and Rezaei et al. \cite{rezaei2024designing} (spacecraft), prove that fixed-time observer–controller synthesis is feasible but have not been extended to the distributed-parameter or flexible-link domain. Inspired by recent progress in fixed-time observer design for finite-dimensional systems such as UAVs and spacecraft, researchers are now exploring similar ideas for distributed-parameter manipulators.



Beyond observer–controller designs, a variety of other strategies have recently emerged to enhance FLM performance. Recent advancements in flexible manipulator control include a linear quadratic optimal boundary control proposed by De Luca et al. \cite{cristofaro2020linear} for SLFMs. Zhu et al. \cite{zhu2023high} presented a two-time-scale adaptive robust control law for tip tracking. Vision-based adaptive control has been explored by Sahu et al. \cite{sahu2023adaptive}. Furthermore, learning-based methodologies, including deep learning \cite{viswanadhapalli2024deep} and reinforcement learning \cite{he2020reinforcement}, have been utilised for tracking control. Reza et al. \cite{mohsenipour2024flexible} introduced a transfer function-based methodology for FLM regulation. Other notable works include boundary disturbance observer-based control for SLFMs subjected to external disturbances, as presented by Zhao et al. \cite{zhao2019boundary}. Sarkhel et al. \cite{sarkhel2023robust} investigated optimal PID controller placement for vibration suppression in rod-type FLMs. Kayastha et al. \cite{kayastha2023comparative} proposed a composite controller that combines predictive control and LQR for post-impact motion regulation of flexible-arm robots. In \cite{sharma2024}, Sharma et al. proposed a functional observer-based output feedback sliding mode control scheme, which employs balanced truncation for model order reduction.  

Despite extensive research, several theoretical and practical challenges remain unresolved. In summary, existing literature provides valuable insight into sliding-mode and observer-based control of flexible manipulators, but several key limitations remain:
\begin{itemize}
\item \textbf{Lack of initial-condition-independent convergence guarantees:} Most available controllers, including NTSMC and fractional-order SMC, ensure finite-time performance only for a given initial condition, not fixed-time regulation with a computable upper bound.

\item \textbf{Absence of fixed-time observer–controller co-design:} Although fixed-time sliding-mode observers have been studied for UAVs and double integrators, no work explicitly integrates such observers with a fixed-time, nonsingular sliding-mode controller for flexible manipulators.

\item \textbf{Singularity and implementation issues:} Terminal-type sliding manifolds often introduce singularities, while high-order or boundary-based PDE controllers require measurements or models that are impractical for low-cost setups.
\end{itemize}
To address the above limitations, this paper proposes a novel fixed-time observer–controller framework for flexible manipulators.
\vspace{-0.4cm}
\subsection{Contributions}
The main contributions are summarised as follows:
\begin{enumerate}
\item The proposed controller is designed using a dynamic model with only the first vibratory mode, but is validated through numerical simulations on a two-mode dynamic model of the SLFM. The observer receives outputs from the two-mode model, which serve as its input. 

\item The proposed scheme bridges the gap between boundary-control theory and practical lumped-parameter implementation by providing a low-order, computationally efficient control structure suitable for experimental validation.

\item A fixed-time sliding-mode observer is designed to reconstruct the unmeasured system dynamics with an explicit convergence-time bound independent of initial conditions.

\item A nested NTSM manifold, in which multiple sliding surfaces are hierarchically embedded within one another, is designed to achieve singularity-free fixed time regulation of the tip position, ensuring global convergence with analytically computable bounds.

\item A rigorous Lyapunov-based stability proof establishes that both observer and controller errors converge in fixed time, guaranteeing overall system stability and robustness against uncertainties.

\item The robustness and effectiveness of the proposed observer–controller framework are demonstrated in the presence of unmodelled higher order dynamics, and its performance is subsequently compared with state-of-the-art control algorithms.

\item The proposed observer–controller framework is further validated through experimental implementation, demonstrating its practical feasibility and effectiveness under real-world operating conditions.
\end{enumerate}

\section{Dynamic Modelling of the SLFM}
Consider an SLFM as depicted in Figure \ref{fig:slfm} of length $l (m)$, which is assumed to have a uniformly distributed mass with linear mass density $\rho (kg/m)$. The following assumptions are taken into consideration while modelling the system: 

\begin{assumption}
    The link is subjected only to small, linear-elastic deformations, primarily bending. Torsional and compressive deformations are neglected.
\end{assumption}
\begin{assumption}
    The effects of gravity and non-linear elastic forces on the link's deformation are considered negligible.
\end{assumption}
\begin{figure}[ht]
    \centering
    \includegraphics[width=0.8\linewidth]{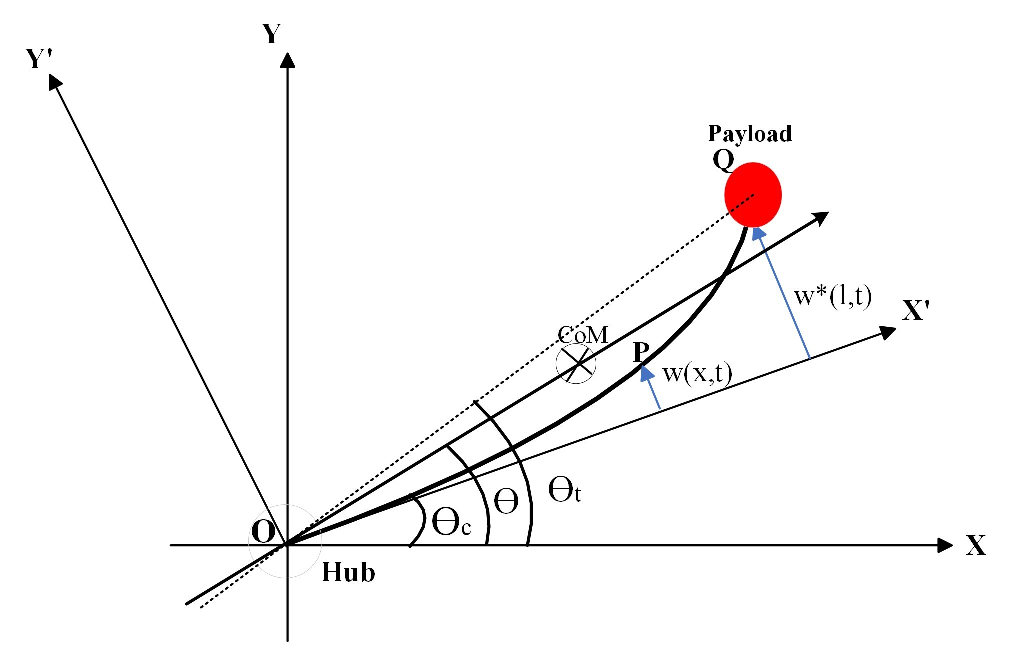}
    \caption{Schematic of Single Link Flexible Manipulator}
    \label{fig:slfm}
\end{figure}

The SLFM is modelled as an Euler-Bernoulli beam with an area moment of inertia $I (m^4)$ and Young's modulus $E (N/m^2)$. The manipulator's payload has a mass $M_p (kg)$ and an inertia $I_p (kg-m^2)$. A motor with inertia $I_h(kg-m^2)$ applies a torque $\tau (N-m)$ to the base.

By utilising Hamilton's principle and the assumed mode method \cite{de2001rest, sharma2024}, the dynamic equations of motion can be expressed in matrix form as:
\begin{subequations}\label{matrix_form}
\begin{align}
       J_t \ddot{\theta}(t) &= \tau(t) \label{rigid_dynamics} \\
     \ddot{p_j}(t) + 2\zeta \omega_j \dot{p_j}(t) + \omega_j^2 p_j(t) &= {\phi}_j'(0) \tau(t) \label{flexible_dynamics} 
\end{align}    
 \end{subequations}
In the above, $J_t$ denotes the total effective rigid-body inertia seen at the hub (including the motor inertia 
$I_h$, payload inertia contribution $I_p$, and the equivalent inertia of the link). Parameters $\omega_j$ and $\zeta$ represent the natural frequency and modal damping ratio of the $j$-th flexible mode, respectively.
 
The angular deviation of the line from the $x$-axis that connects the base to the centre of mass (CoM) is denoted by $'\theta(t)'$. The vibrational behaviour of the system is captured by oscillatory modes $p_j(t)$, where each $p_j(t)$ corresponds to the $j$-th oscillatory mode.

In the conventional state-space framework, the dynamic model represented by equation \eqref{matrix_form} can be represented as:
\begin{subequations} \label{SYSTEM}
\begin{align}
\dot{\psi}(t) &= \mathcal{A} \psi(t) + \mathcal{B} u(t) + h(\psi,t)\label{state_space_model} \\
g(t) &= \mathcal{C}\psi(t) \label{output}
\end{align}
\end{subequations}
\begin{align*}
\psi(t) &= \begin{bmatrix}
    \theta(t) &
    \dot{\theta}(t) &
    p_1(t) & 
    \dot{p}_1(t) &
    \ldots &
    p_n(t) &
    \dot{p}_n(t)
\end{bmatrix}^T \\ 
g(t) &= \begin{bmatrix}
    \theta_t(t) &
    \theta_c(t)
\end{bmatrix}^T
\end{align*}
where $\psi(t) \in \mathbb{R}^{m} \hspace{0.1cm} (m = 2n+2)$ represents internal states of the system, $\mathcal{A} \in \mathbb{R}^{m\times m}$, $ \mathcal{B} \in \mathbb{R}^{m\times 1}$, $\mathcal{C} \in \mathbb{R}^{2 \times m}$ are the system matrices represented below, the lumped matched uncertainties due to parameter variations, unmodelled dynamics, and external disruptions are represented by $h(\psi,t) \in \mathbb{R}^{m}$, $g(t) \in \mathbb{R}^{2}$ is the output of the system, $u(t) \in \mathbb{R}$ denotes the input to the system, and $\theta_t(t)$ and $\theta_c(t)$ are the tip position angle and clamped joint angle, respectively.

The measurable outputs are the tip pointing angle $\theta_t(t)$ and the clamped-joint angle $\theta_c(t)$. Under the small-deflection assumption, the tip angle is approximated by the small-angle relation
\begin{align*}
    \theta_t(t) = \theta(t) + \frac{w(l,t)}{l}
\end{align*}
while the clamped-joint angle corresponds to the base slope,
\begin{align*}
    \theta_c(t) = \theta(t) + w(0,t)
\end{align*}
These relations determine the structure of the output matrix $\mathcal{C}$.
\begin{align*}
\mathcal{A} &= \begin{bmatrix}
0 & 1 & 0 & 0 & \cdots & 0 & 0 & 0 & 0 \\
0 & 0 & 0 & 0 & \cdots & 0 & 0 & 0 & 0\\
0 & 0 & 0 & 1 & \cdots & 0 & 0 & 0 & 0\\
0 & 0 & -\omega_1^2 & -2\zeta\omega_1 & \cdots & 0 & 0 & 0 & 0 \\
\vdots & \vdots & \vdots & \vdots & \ddots & \vdots & \vdots & \vdots & \vdots \\
0 & 0 & 0 & 0 & \cdots & 0 & 0 & 0 & 1 \\
0 & 0 & 0 & 0 & \cdots & 0 & 0 & -\omega_n^2 & -2\zeta\omega_n
\end{bmatrix} \\
\mathcal{B} &= \begin{bmatrix}
0 & \frac{1}{J_t} & 0 & \phi_1'(0) & \cdots & 0 & \phi_n'(0)
\end{bmatrix}^T  \\ 
\mathcal{C} &= \begin{bmatrix}
1 & 0 & \frac{\phi_1(l)}{l} & 0 & \cdots & \frac{\phi_n(l)}{l} & 0 \\
1 & 0 & \phi_1'(0) & 0 & \cdots & \phi_n'(0) & 0
\end{bmatrix} 
\end{align*}
The following are assumed to be true while designing the control input for the system in \eqref{SYSTEM}.
\begin{assumption}
The lumped uncertainties and disturbances are assumed to be matched with the system input. Thus, there exists a vector $\xi(\psi,t) \in \mathbb{R}$ such that the uncertainty term $h(\psi,t) \in \mathbb{R}^{m}$ can be expressed as $\mathcal{B}\xi(\psi,t) \in \mathbb{R}^{m}$. \label{assm3}
\end{assumption}
\begin{assumption}
The magnitude of the uncertainty vector $\xi(\psi,t)$ is bounded by a known constant $\Bar{\xi}$. That is $||\xi(\psi,t)||\leq \Bar{\xi}, \forall t\ge 0$ \label{assm4}
\end{assumption}

\begin{remark}
Assumptions~\ref{assm3}–\ref{assm4} are standard in robust and sliding mode control literature. The matched uncertainty condition ensures that disturbances enter through the same channel as the control input, allowing effective compensation. The boundedness of the uncertainty vector enables proper selection of control gains and facilitates finite-time stability analysis. These assumptions are practically reasonable and often validated through prior modelling or experimental estimation.
\end{remark}

\subsection{Control Objective}
The control objective is to design a robust controller that ensures the tip position of the flexible link, $\theta_t(t)$, faithfully achieves a predefined desired position, $\theta_d$, in a fixed amount of time. Mathematically, this can be expressed as minimising the tracking error $\Tilde{\theta}(t)$, defined as the difference between the actual and desired tip positions, such that
\begin{align}
     \Tilde{\theta}(t) = \theta_t(t) - \theta_d = 0 \hspace{0.25cm}\forall \hspace{0.25cm} t \geq T \label{relative_error}
\end{align}
To ensure that the controller focuses on achieving the desired tip position without unintentionally exciting the other vibratory modes, the desired position in the state-space coordinates, $\psi_d(t)$, is chosen as:
\begin{align*}
    \psi_d(t) = \begin{bmatrix}
        \theta_d & 0 & 0 & 0 & \cdots & 0 & 0 
    \end{bmatrix}
\end{align*}

\subsection{State Transformation for Fixed-Time SMC Based Controller Design}
In order to facilitate the design of a fixed-time sliding mode controller with guaranteed convergence properties, the original system in \eqref{SYSTEM} is first transformed into its controllable canonical form (or chain of integrators representation) \cite{chen1984linear}. This transformation explicitly exposes the controllability structure of the system and renders the input–state relationship in a form that is particularly convenient for fixed-time and terminal sliding mode control design. In this canonical representation, the control input directly affects the highest-order state, while the remaining states evolve as a chain of successive integrators. This property significantly simplifies both the construction of the sliding manifold and the subsequent Lyapunov-based stability analysis.

Applying the controllable canonical transformation $\mathbf{z}(t) = \mathcal{T}\psi(t)$, where $\mathcal{T}$ is the transformation matrix, the system can be rewritten in the transformed form as:
\begin{subequations} \label{transformed_system}
    \begin{align}
        \dot{\mathbf{z}}(t) &= \mathcal{A}_C \mathbf{z}(t) + \mathcal{B}_Cu + \delta(\mathbf{z},t) \\
        g(t) &= \mathcal{C}_C\mathbf{z}(t)
    \end{align}
\end{subequations}
where $\mathbf{z}(t)$ is the transformed state of the system, $\mathcal{A}_C = \mathcal{T} \mathcal{A} \mathcal{T}^{-1}$, $\mathcal{B}_C = \mathcal{T}\mathcal{B}$ and $\mathcal{C}_C = \mathcal{C}\mathcal{T}^{-1}$ are the transformed system matrices.
\begin{align*}
\mathcal{A}_C &= \begin{bmatrix}
    0 & 1 & 0 & \cdots & 0 \\
    0 & 0 & 1 & \cdots & 0 \\
    \vdots & \vdots & \vdots & \ddots & \vdots \\
    0 & 0 & 0 & \cdots & 1 \\
    -f_0 & -f_1 & -f_2 & \cdots & -f_{m-1}
\end{bmatrix} \\
\mathcal{B}_C &= \begin{bmatrix}
    0 & 0 & 0 & \cdots & 1
\end{bmatrix}^T
\end{align*}
where the coefficients $f_0, f_1, \ldots, f_{m-1}$ correspond to the characteristic polynomial of the system. In this form, the dynamics explicitly appear as a chain of integrators driven by the control input $u(t)$.

The desired position $\mathbf{z}_d$, in the transformed coordinates can be expressed as follows:
\begin{align}
  \mathbf{z}_d &= \mathcal{T}\psi_d  \label{transformed_desired_trajectory}
\end{align}
where $\psi_d$ represents the desired position in the original coordinates.

This coordinate transformation not only simplifies the synthesis of fixed-time sliding mode control but also provides a unified framework for developing non-singular, analytically verifiable proofs of fixed-time convergence. By expressing the system in controllable canonical form, the controller design can exploit the structure of integrator-type dynamics to ensure fixed-time reaching and fixed-time stability with minimal structural assumptions on the original system matrices.

\section{Fixed Time Composite Controller Design}
\subsection{Nested Non-Singular Terminal Sliding Mode Controller} \label{sec:controller}
It has been well established that introducing non-smooth dynamics can significantly enhance system performance \cite{yu2020terminal}. For example, the incorporation of a terminal attractor, a first-order dynamic with a fractional power term, into neural learning rules has been shown to accelerate adaptation and achieve finite/fixed time convergence, as demonstrated in \cite{yu2020terminal}. Based on this, we propose a nested non-singular terminal sliding mode (NNTSM)-based control design for the SLFM system to ensure fixed-time convergence of the tip to its desired position.

To rigorously establish the fixed-time stability properties of the proposed approach, the following subsection presents the fundamental fixed-time stability lemma, which serves as the theoretical basis for the subsequent controller design and analysis.

\begin{lemma}(\textit{Fixed Time Lemma} \cite{polyakov2011nonlinear})\label{fixed_time_lemma}
Consider a continuous autonomous system:
\begin{align}
    \dot{x}(t) = f(x), \hspace{0.5cm} x(0) = 0, \hspace{0.5cm} f(0) = 0
\end{align}
Assume there exist a Lyapunov function $V(x):\mathbb{R}^n \rightarrow \mathbb{R}^+ \cup \{0\}$ such that $\dot{V} \leq -\eta V^{\lambda_1}(x) - \upsilon V^{\lambda_2}(x)$ holds for positive constants $\eta,~\upsilon,~\lambda_1 \in (0,1)$, and $\lambda_2 > 1$.

Then, the origin is globally fixed-time stable, and the settling time $T_s$ of the state x(t) is bounded by:
\begin{align}
T_s < T_{max.} = \frac{1}{\eta(1-\lambda_1)} + \frac{1}{\nu (\lambda_2 - 1)}
\end{align}
\end{lemma}
\subsubsection{Preliminaries and Control Law Design}
For controller design, only the first vibratory mode is considered in this paper. Consequently, the resulting state space model is of fourth order and can be represented as follows:
\begin{equation}\label{eq:system}
\dot{z}_1=z_2,\quad 
\dot{z}_2=z_3,\quad 
\dot{z}_3=z_4,\quad 
\dot{z}_4=f(z)+u+\delta(t,z),
\end{equation}
where $z=[z_1,z_2,z_3,z_4]^\top\!\in\!\mathbb{R}^4$, the smooth nonlinear function $f(z)$ is assumed known (to be estimated later by an observer), and the disturbance $\delta(t,z)$ is bounded by $|\delta(t,z)|\le \bar d$.  
Let the desired position $z_{d_1}$ and its derivatives be bounded, and define the position error $e = z_1 - z_{d_1}$.

\noindent\textbf{Definition 1 (Non-singular auxiliary function).}\label{def:1}

To avoid singularities, we employ the smooth sub-unity map
\begin{equation}
\phi_\varepsilon(x) := x(x^2+\varepsilon^2)^{\frac{\gamma_1-1}{2}}, 
\qquad 0<\gamma_1<1,\ \varepsilon>0,
\label{eq:phi_eps}
\end{equation}
which is $C^1$ for all $x$ and satisfies, for all $x\in\mathbb{R}$,
\begin{equation}
\varepsilon^{\gamma_1-1}x^2 
\;\le\;
x\phi_\varepsilon(x)
\;\le\;
|x|^{\gamma_1+1}.
\label{eq:phi_bounds}
\end{equation}
Thus $\phi_\varepsilon(x)$ behaves like $|x|^{\gamma_1}\mathrm{sgn}(x)$ away from the origin and is linear near $x=0$.

Fix $\gamma_2>1$ and choose positive design gains $\alpha_i>0$, $\kappa_{i1}>0$, $\kappa_{i2}>0$ for $i=0,1,2$. 

\noindent\textbf{Definition 2 (Sliding surfaces).}
Let $\alpha_i,\kappa_{i1},\kappa_{i2}>0$ and $0<\gamma_1<1<\gamma_2$.
Define the recursive surfaces
\begin{align}\label{eq:surfaces}
s_0 &= e, \nonumber\\
s_{i+1} &= \dot{s}_i + \alpha_i s_i
          + \kappa_{i1}\phi_\varepsilon(s_i)
          + \kappa_{i2}|s_i|^{\gamma_2}\mathrm{sgn}(s_i),
          \quad i=0,1,2.
\end{align}

Differentiation yields the dynamics
\begin{align}
\dot{s}_i &= -\alpha_i s_i - \kappa_{i1}\phi_\varepsilon(s_i) - \kappa_{i2}|s_i|^{\gamma_2}\mathrm{sgn}(s_i) + s_{i+1},~ i=0,1,2.\label{eq:sidot}
\\
\dot{s}_3 &= f(z)+u+\Phi(t,z)+\delta(t,z). \label{eq:s3dot} 
\end{align}
where $\Phi(t,z)$ collects known feed-forward and reference‐trajectory terms.

The control input is designed as
\begin{multline}\label{eq:control}
u = -f(z) - \Phi(t,z)
    -c_1|s_3|^{p}\mathrm{sgn}(s_3) \\
    -c_2|s_3|^{q}\mathrm{sgn}(s_3) 
    -\eta\,\mathrm{sgn}(s_3),    
\end{multline}
with constants $c_1,c_2,\eta>0$ and exponents $0<p<1<q$, $\eta>\bar d$.

\subsubsection{Main Result}

\begin{theorem}\label{thm:fixed_time_thm}
Consider the transformed SLFM dynamics in \eqref{transformed_system} and the sliding variables in \eqref{eq:sidot}. Under the control law \eqref{eq:control}, the closed-loop system is globally practically fixed-time stable. That is, there exist constants $T_{\mathrm{ctrl}}>0$ and $\varepsilon_i>0$, $i=0,1,2,3$, independent of the initial conditions, such that
\begin{align}
|s_i(t)| \le \varepsilon_i, \qquad i=0,1,2,3, \qquad \forall t \ge T_{\mathrm{ctrl}}.
\end{align}
In particular,
\begin{align}
|e(t)| \le \varepsilon_0, \qquad \forall t \ge T_{\mathrm{ctrl}},
\end{align}
where $e(t)=z_1(t)-z_{d1}(t)$. Furthermore, $T_{\mathrm{ctrl}} \le T_0+T_1+T_2+T_3,$ with $T_3 \le \frac{1}{c_1(1-p)}+\frac{1}{c_2(q-1)},$ and
\begin{align}
T_i &\le \frac{1}{\underline{\kappa}_{i1}(1-\gamma_1)}+\frac{1}{\underline{\kappa}_{i2}(\gamma_2-1)}, \qquad i=0,1,2.
\end{align}
The constants $\varepsilon_i$ can be made arbitrarily small by appropriate choice of the controller parameters, particularly by reducing the regularization parameter $\varepsilon$.
\end{theorem}

\proof
The proof follows a recursive fixed-time analysis of the surfaces
$s_3\to s_2\to s_1\to s_0$.

\medskip
\noindent
\textit{Step~1 (Surface $s_3$):}
Substituting~\eqref{eq:control} into~\eqref{eq:s3dot} gives
\[
\dot{s}_3
= -c_1|s_3|^{p}\mathrm{sgn}(s_3)
  -c_2|s_3|^{q}\mathrm{sgn}(s_3)
  +(\delta-\eta\,\mathrm{sgn}(s_3)).
\]
Since $|\delta|\le\bar d<\eta$, it follows that
\[
s_3\dot{s}_3 \le -c_1|s_3|^{p}-c_2|s_3|^{q}.
\]
Let $V_3=\tfrac12 s_3^2$. Then,
\[
\dot V_3 \le -a_3V_3^{\alpha_3}-b_3V_3^{\beta_3},
\]
where $a_3=c_1 2^{-\frac{p+1}{2}}$, $b_3=c_2 2^{-\frac{q+1}{2}}$,
$\alpha_3=\tfrac{p+1}{2}\!\in(0,1)$, and $\beta_3=\tfrac{q+1}{2}\!>\!1$.
By the fixed-time stability lemma,
\[
T_3 \le \tfrac{1}{a_3(1-\alpha_3)}+\tfrac{1}{b_3(\beta_3-1)}.
\]
Hence,
\[
s_3(t)=0,\qquad \forall~ t\ge T_3.
\]

\medskip
\noindent
\textit{Step 2: Practical fixed-time convergence of $s_i$, $i=0,1,2$.}
Once $s_{i+1}$ has converged to a sufficiently small neighbourhood of the origin, the dynamics of $s_i$ can be expressed as
\begin{align}
\dot{s}_i=-\alpha_i s_i-\kappa_{i1}\phi_\varepsilon(s_i)-\kappa_{i2}|s_i|^{\gamma_2}\operatorname{sgn}(s_i)+d_i(t),
\end{align}
where $|d_i(t)|\le \bar d_i$ after a fixed time. For the Lyapunov function $V_i=\frac{1}{2}s_i^2$, one has
\begin{align}
\dot{V}_i
=-\alpha_i s_i^2-\kappa_{i1}s_i\phi_\varepsilon(s_i)-\kappa_{i2}|s_i|^{\gamma_2+1}+s_i d_i(t).
\end{align}
Since $s_i\phi_\varepsilon(s_i)\ge c_\varepsilon |s_i|^{\gamma_1+1}$ for some $c_\varepsilon>0$, Young's inequality gives
\begin{align}
\dot{V}_i \le -a_i V_i^{\frac{\gamma_1+1}{2}}-b_i V_i^{\frac{\gamma_2+1}{2}}+\sigma_i,
\end{align}
with $a_i,b_i>0$ and $\sigma_i\ge 0$. Therefore, $s_i(t)$ reaches the residual set
\begin{align}
\Omega_i=\{\,s_i\in\mathbb{R}:|s_i|\le \varepsilon_i\,\}
\end{align}
in a fixed time $T_i$ independent of the initial conditions.
\medskip
\noindent
\textit{Step~3 (Total settling time)}:
By recursive application, all surfaces $s_{i}(t)$, $i=0,1,2,3$, converge to their respective residual sets in fixed time. In particular, the tracking error $e(t) = s_0$ satisfies $|e(t)| \leq \varepsilon$ after a fixed time. The overall settling time is bounded as fixed time,
\begin{align}\label{eq:T_settling_control}
T_{\mathrm{controller}} \leq \sum_{i=0}^{3} T_i
\end{align}
which is independent of initial conditions. Moreover, the residual sets can be made arbitrarily small by choosing $\varepsilon>0$ sufficiently small. This completes the proof.
\endproof

We can perform a generalised proof of this theorem for a larger number of vibratory modes, considering an appropriate number of nested sliding functions.
\vspace{-0.15cm}
\subsection{Fixed-Time Sliding-Mode Observer}\label{subsec:ft_smo} 
The fixed-time controller in Sec.~\ref {sec:controller} requires knowledge of the transformed system states $z(t)$ obtained from the first vibratory mode model. However, the available measurements originate from the actual two-mode plant. Let the plant output be
\begin{equation}
    y = C_2 z_2 = C_r z + \Delta y,
    \label{eq:two_mode_output}
\end{equation}
where $\Delta y$ denotes the residual contribution of the higher vibratory mode, bounded as $\|\Delta y\|\le \bar{\Delta}$.

Our objective is to reconstruct the reduced transformed state $z(t)$ from the available measurement $y(t)$ and input $u(t)$, such that the controller in Sec.~\ref{sec:controller} can employ $\hat{z}(t)$ with guaranteed fixed-time convergence of the estimation error.

\begin{theorem}\label{theorem_2}
Consider an observer of the following form for the system in \eqref{transformed_system}.
\begin{multline}\label{observer_dynamics}
       \dot{\hat{\mathbf{z}}} = \mathcal{A}_C \hat{\mathbf{z}}(t) + \mathcal{B}_C u(t) + \mathcal{L} (g - C_c \hat{z}(t)) \\+ \mathcal{K}_1 {\lfloor \Tilde{e}_y \rceil^{\mu_1}}
       + \mathcal{K}_2  {\lfloor \Tilde{e}_y \rceil^{\mu_2}}  
\end{multline} 
where ${\lfloor \Tilde{e}_y \rceil^{\mu_i}} = \begin{bmatrix}
       |\Tilde{e}_{y_1}|^{\mu_i} \text{sgn}(\Tilde{e}_{y_1}) \\
       |\Tilde{e}_{y_2}|^{\mu_i} \text{sgn}(\Tilde{e}_{y_2}) 
   \end{bmatrix}$, for $(i= 1,2)$, with $\mathcal{K}_1 \in \mathbb{R}^{m\times 2}$, $\mathcal{K}_2 \in \mathbb{R}^{m\times 2} $,  constant observer gain matrices of compatible dimensions, chosen such that the nonlinear output-injection terms dominate the estimation error dynamics, $0<\mu_1<1$ and $\mu_2>1$, $\mathcal{L}$ is a design constant, and $\Tilde{e}_y = g - \mathcal{C}_C\hat{\mathbf{z}}$ is the output estimation error. The observer in (\ref{observer_dynamics}) will take the estimation error $\Tilde{e} = \mathbf{z} - \hat{\mathbf{z}}$ to the origin in a fixed amount of time.
\end{theorem}
\proof
Let us choose a Lyapunov function as:
\begin{align}
{V}(t) &= \frac{1}{2} \Tilde{e}^T (\mathcal{C}_C^T\mathcal{C}_C)\Tilde{e} \label{lyap_observer} = \frac{1}{2}\Tilde{e}^T_y \Tilde{e}_y ~~~\quad (\Tilde{e}_y = \mathcal{C}_C\Tilde{e})
\end{align}
The time derivative of output estimation error $\Tilde{e}_y$ using equations \eqref{transformed_system} and \eqref{observer_dynamics} can be expressed as:
\begin{multline}
\dot{\Tilde{e}}_y = \mathcal{C}_C (\mathcal{A}_C z(t) + \mathcal{B}_Cu(t) - \mathcal{A}_C \hat{z}(t) - \mathcal{B}_Cu(t) - \mathcal{L}\mathcal{C}_C e \\
- \mathcal{K}_1 |\Tilde{e}_y|^{\mu_1} \text{sgn}(\Tilde{e}_y)
       - \mathcal{K}_2 |\Tilde{e}_y|^{\mu_2} \text{sgn}(\Tilde{e}_y)) \label{ey_dot}
\end{multline}
Thus,
\begin{align}
    \dot{{V}}(t) &=\frac{1}{2}\Tilde{e}^T_y \dot{\Tilde{e}}_y + \frac{1}{2}\dot{\Tilde{e}}^T_y \Tilde{e}_y \label{lyap_dot_obs} \\
    &= \frac{1}{2}\left[\Tilde{e}^T(\mathcal{C}^T\mathcal{C} (\mathcal{A}_C - \mathcal{L}\mathcal{C}_C) +  (\mathcal{A}_C - \mathcal{L}\mathcal{C}_C)^T \mathcal{C}^T\mathcal{C}) \Tilde{e}\right] \nonumber \\ 
    &- \Tilde{e}^T \mathcal{C}^T\mathcal{C} \mathcal{K}_1 {\lfloor \Tilde{e}_y \rceil^{\mu_1}} - \Tilde{e}^T \mathcal{C}^T\mathcal{C} \mathcal{K}_2 {\lfloor \Tilde{e}_y \rceil^{\mu_2}} \label{V_d1} 
\end{align}
where $\left(\mathcal{C}^T\mathcal{C} (\mathcal{A}_C - \mathcal{L}\mathcal{C}_C) +  (\mathcal{A}_C - \mathcal{L}\mathcal{C}_C)^T \mathcal{C}^T\mathcal{C}\right) = - Q$, here $Q$ is a positive definite matrix. Therefore,
\begin{align*}
\dot{{V}}(t) =& - \frac{1}{2} \Tilde{e}^TQ\Tilde{e} -\Tilde{e}_y^T \mathcal{C} \mathcal{K}_1 \begin{bmatrix}
       |\Tilde{e}_{y_1}|^{\mu_1} \text{sgn}(\Tilde{e}_{y_1}) \\
       |\Tilde{e}_{y_2}|^{\mu_1} \text{sgn}(\Tilde{e}_{y_2}) 
   \end{bmatrix} \\ 
   & -  \Tilde{e}_y^T \mathcal{C} \mathcal{K}_2 \begin{bmatrix}
       |\Tilde{e}_{y_1}|^{\mu_2} \text{sgn}(\Tilde{e}_{y_1}) \\
       |\Tilde{e}_{y_2}|^{\mu_2} \text{sgn}(\Tilde{e}_{y_2}) 
   \end{bmatrix} \\
    \dot{{V}}(t) \leq &- \|\mathcal{C}_C\| \|\mathcal{K}_1\| \begin{bmatrix}
        \Tilde{e}_{y_1} &  \Tilde{e}_{y_2} 
    \end{bmatrix} \begin{bmatrix}
       |\Tilde{e}_{y_1}|^{\mu_1} \text{sgn}(\Tilde{e}_{y_1}) \\
       |\Tilde{e}_{y_2}|^{\mu_1} \text{sgn}(\Tilde{e}_{y_2}) 
   \end{bmatrix}   \\ 
    &- \|\mathcal{C}_C\| \|\mathcal{K}_2\| \begin{bmatrix}
        \Tilde{e}_{y_1} &  \Tilde{e}_{y_2}
    \end{bmatrix} \begin{bmatrix}
       |\Tilde{e}_{y_1}|^{\mu_2} \text{sgn}(\Tilde{e}_{y_1}) \\
       |\Tilde{e}_{y_2}|^{\mu_2} \text{sgn}(\Tilde{e}_{y_2}) 
   \end{bmatrix} \\
    \leq & - \|\mathcal{C}_C\| \|\mathcal{K}_1\| \sum_{i=1}^{2} |\Tilde{e}_{y_i}|^{(\mu_1+1)} -  \|\mathcal{C}_C\| \|\mathcal{K}_2\| \sum_{i=1}^{2} |\Tilde{e}_{y_i}|^{(\mu_2+1)} 
\end{align*}
\begin{align}
\dot{{V}}(t)  \leq & - \|\mathcal{C}_C\| \|\mathcal{K}_1\| \left(\sum_{i=1}^{2} |\Tilde{e}_{y_i}|^2\right)^{\left(\frac{\mu_1+1}{2}\right)} \nonumber \\
&  -  \|\mathcal{C}_C\| \|\mathcal{K}_2\| \left(\sum_{i=1}^{2} |\Tilde{e}_{y_i}|^2\right)^{\left(\frac{\mu_2+1}{2}\right)}  \label{V2_dot_final}  
\end{align} 
From equation \eqref{lyap_observer} and \eqref{V2_dot_final},
\begin{align}
 \dot{{V}}(t) \leq &- \varepsilon_1 {V}^{\left(\frac{\mu_1+1}{2}\right)} - \varepsilon_2 {V}^{\left(\frac{\mu_2+1}{2}\right)} \label{V2_d} 
\end{align}    
where $\varepsilon_1 = \|\mathcal{C}_C\| \|\mathcal{K}_1\| $ and $\varepsilon_2 = \|\mathcal{C}_C\| \|\mathcal{K}_2\|$ are both positive definite. 

Define
\[
\rho_1 =\frac{\mu_1+1}{2}\in(0,1),\qquad
\rho_2 =\frac{\mu_2+1}{2}>1.
\]
Then, by the fixed-time lemma \ref{fixed_time_lemma} for scalar systems, the origin is fixed-time stable and the settling time is uniformly bounded by
\begin{equation}
T_{FTSMO} \le \frac{1}{\varepsilon_1(1-\rho_1)}+\frac{1}{\varepsilon_2(\rho_2-1)}.
\label{eq:T_bound_observer}
\end{equation}
Substituting $\rho_1$ and $\rho_2$ yields the explicit bound
\begin{equation}
T_{FTSMO}
\le \frac{2}{\varepsilon_1(1-\mu_1)}
   +\frac{2}{\varepsilon_2(\mu_2-1)}.
\label{eq:T_bound_explicit}
\end{equation}

By using the fixed time lemma \ref{fixed_time_lemma} in equation \eqref{V2_dot_final}, it can be concluded that the Lyapunov function ${V}(t)$ will converge to the origin in a fixed amount of time. Further, the output estimation error $\Tilde{e}_y$ will also converge to zero in a fixed amount of time. 
\[
\Tilde{e}_y(t) = \mathcal{C}_C \Tilde{e}(t) \equiv 0, \quad \forall t \ge T_{FTSMO}.
\]
Then all time derivatives of $\Tilde{e}_y(t)$ are also zero for $t \ge T_{FTSMO}$, i.e.,
\[
\Tilde{e}_y^{(k)}(t) = \mathcal{C}_C \mathcal{A}_C^k \Tilde{e}(t) \equiv 0, \quad \forall k = 0,1,\dots,n-1,\ \forall t \ge T_{FTSMO}.
\]
Evaluating at $t = T_{FTSMO}$, we obtain
\[
\begin{bmatrix}
\mathcal{C}_C &
\mathcal{C}_C \mathcal{A}_C &
\cdots &
\mathcal{C}_C \mathcal{A}_C^{n-1}
\end{bmatrix}^\top
\Tilde{e}(T_{FTSMO}) = 0.
\]
Since the pair $(\mathcal{A}_C,\mathcal{C}_C)$ is observable, the observability matrix
\[
\mathcal{O} =
\begin{bmatrix}
\mathcal{C}_C &
\mathcal{C}_C \mathcal{A}_C &
\cdots &
\mathcal{C}_C \mathcal{A}_C^{n-1}
\end{bmatrix}^\top
\]
has full rank, which implies that $\Tilde{e}(T_{FTSMO}) = 0$. By uniqueness of solutions to the system $\dot{\Tilde{e}}(T_{FTSMO}) = \mathcal{A}_C \Tilde{e}$, it follows that $\Tilde{e}(t) \equiv 0$ for all $t \ge T_{FTSMO}$. Therefore, the state error converges to zero in a fixed time. This completes the proof.
\endproof
\begin{remark}[Integration with the controller]
The observer \eqref{observer_dynamics} provides estimates $\hat{z}$ used by the fixed-time controller of Sec.~\ref{sec:controller}. Since the observer error $e_z$ converges within $T_{\text{FTSMO}}$, which is independent of initial conditions, the overall closed loop exhibits a cascaded fixed-time property.  
By designing $L_2$ such that $T_{\text{FTSMO}} < T_{\text{ctrl}}$, 
the controller performance remains unaffected by the transient of the observer.
\end{remark}
\vspace{-0.25cm}
\subsection{Composite Controller Design}
\label{subsec:composite_controller}
The previous sections established two key components of the overall control architecture:
\begin{itemize}
    \item[(i)] A \emph{non-singular fixed-time controller} (Theorem~\ref{thm:fixed_time_thm})
    ensuring finite and uniform convergence of the transformed tracking error \( e = z_1 - z_{d_1} \).
    \item[(ii)] A \emph{fixed-time sliding-mode observer} (Theorem~\ref{theorem_2})
    guaranteeing that the transformed state estimation error \( e_z = \hat{z} - z \)
    vanishes within a bounded fixed time \( T_{\text{FTSMO}} \), independently of initial conditions.
\end{itemize}

To realise a fully implementable scheme, the controller in~\eqref{eq:control} is modified to use
the estimated transformed states \(\hat{z}\) obtained from the observer~\eqref{observer_dynamics}.
The resulting composite closed-loop controller is given by
\begin{multline}\label{eq:composite_u} 
u = -f(\hat{z}) - \Phi(t,\hat{z})
    - c_1|s_3(\hat{z})|^{p}\mathrm{sgn}(s_3(\hat{z}))
   \\ - c_2|s_3(\hat{z})|^{q}\mathrm{sgn}(s_3(\hat{z}))
    - \eta\,\mathrm{sgn}(s_3(\hat{z})), 
\end{multline}

where \(s_3(\hat{z})\) denotes the terminal sliding variable computed using the observer states according to~\eqref{eq:surfaces}, and the auxiliary function \(\Phi(t,\hat{z})\) is constructed from the desired trajectory and its derivatives.

\vspace{2mm}
\noindent\textbf{Closed-loop structure:}
The overall system can be viewed as a cascade of two fixed-time subsystems:
\begin{align*}
&\text{(observer subsystem)} \;\; e_z \;\xrightarrow[\text{FTS in } T_{\text{FTSMO}}]{}\; 0,\\
&\text{(controller subsystem)} \;\; e \;\xrightarrow[\text{FTS in } T_{\text{ctrl}}]{}\; 0.
\end{align*}
From Theorems~\ref{thm:fixed_time_thm} and~\ref{theorem_2}, both subsystems are globally fixed-time stable, with respective settling-time bounds in \eqref{eq:T_settling_control} and \eqref{eq:T_bound_observer}.

If the observer gains are selected such that
\begin{equation}
    T_{\text{FTSMO}} < T_{\text{ctrl}},
    \label{eq:time_hierarchy}
\end{equation}
Then the composite closed loop exhibits \emph{cascaded fixed-time stability}, and both the estimated and actual tracking errors converge to zero within a total settling time bounded by
\begin{equation}
    T_{\text{total}} \le T_{\text{FTSMO}} + T_{\text{ctrl}}.
\end{equation}

\begin{remark}[Composite closed-loop stability]
Condition~\eqref{eq:time_hierarchy} ensures a strict time-scale separation between estimation and control. Hence, the state estimates \(\hat{z}\) converge before the controller’s surfaces \(s_i(\hat{z})\)
reach the origin. Under this hierarchy, the overall composite system is globally fixed-time stable in the sense of Polyakov \cite{polyakov2011nonlinear}, and the final control input~\eqref{eq:composite_u} can be implemented directly using measured output \(y\) and known model quantities.
\end{remark}

\section{Simulation Results}
The proposed composite fixed-time controller in \eqref{eq:composite_u} for the SLFM system is validated using numerical simulation. The proposed controller ensures convergence of the tip position to the desired position within a fixed time. The simulations were conducted using MATLAB/Simulink (R2018a) on a 64-bit Windows system with an Intel Core i7 processor and 16 GB of RAM. Table \ref{tab:flexible_link_param} gives the physical parameters for numerical simulation. The desired tip position is $\frac{\pi}{4}$. 

\begin{table}[ht]
\centering
\begin{tabular}{|c|c|c|c|}
\hline
\textbf{Parameter}  & \textbf{Value} & \textbf{Parameter}  & \textbf{Value}   \\
\hline
$\rho$ & 0.5  $kg/m$ & $\phi'_1(0)$ & 32.8184 \\
l & 1  $m$ & $\omega_1$ & 20.53 $rad/sec$   \\
$m_p$ & 0  $kg$ &  $J_p$ & 0 $kg-m^2$   \\
$J_0$ & 0.002  $kg-m^2$ & $\phi_1(l)$ & 0.3214   \\
$EI$ & 1  $N-m^2$ &  & \\
\hline
\end{tabular}
\caption{Physical Parameters of SLFM}
\label{tab:flexible_link_param}
\end{table}
Fig.~\ref{fig:fx_time_o/p} reports the time histories of the tip angle $\theta_t(t)$ (top subplot) and the joint angle 
$\theta_c(t)$ (bottom subplot) of the SLFM under the proposed controller, FOSMC \cite{sharma2024position}, and a PD controller~\cite{zollo2007pd}. The desired position is indicated by the dotted magenta line. As seen in the tip response, the proposed scheme achieves a smooth rise and accurate regulation, with negligible overshoot (below $\approx2\%$) and a settling time of about $2.5-3~sec$. In contrast, FOSMC exhibits a pronounced transient overshoot (exceeding $\approx40\%$) and oscillatory convergence, whereas the PD controller shows slower convergence accompanied by sustained oscillations around the reference.

The joint angle trajectories further corroborate these observations. The proposed controller yields a comparatively well-damped joint response with reduced oscillation amplitude after approximately $3~sec$, indicating effective suppression of flexible-mode excitation. Conversely, the FOSMC response contains large-amplitude oscillations and sharp excursions during the transient, while the PD controller retains persistent oscillatory behaviour. Collectively, these results suggest that the proposed design improves tracking performance at both the joint and tip levels while mitigating vibration-induced oscillations, consistent with the stability analysis.
\begin{figure}
    \centering
    \includegraphics[width= 8cm, height = 4cm]{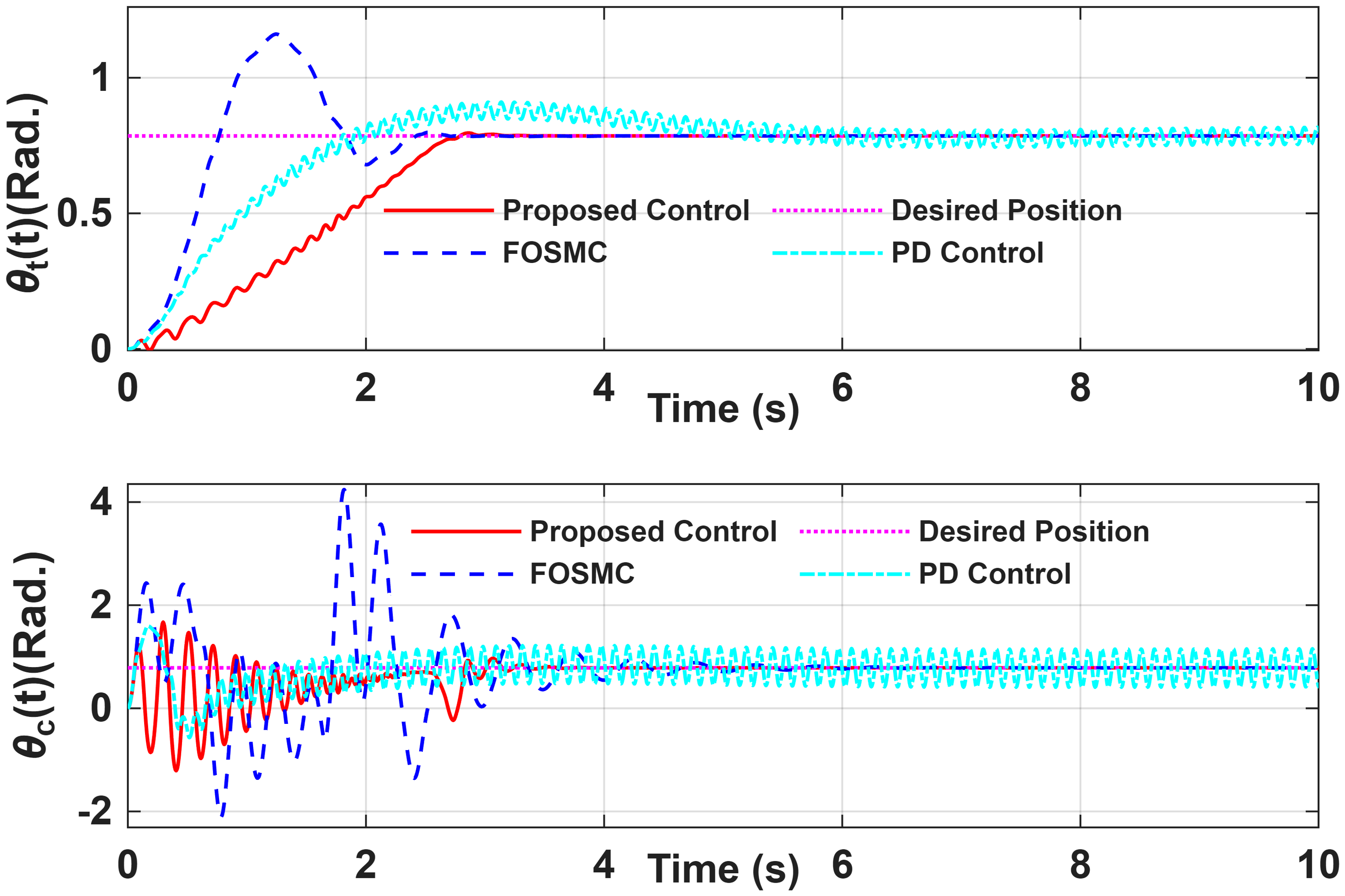}
    \caption{Output Plots}
    \label{fig:fx_time_o/p}
\end{figure} 
\begin{figure}
    \centering
    \includegraphics[width= 8cm, height = 4.5cm]{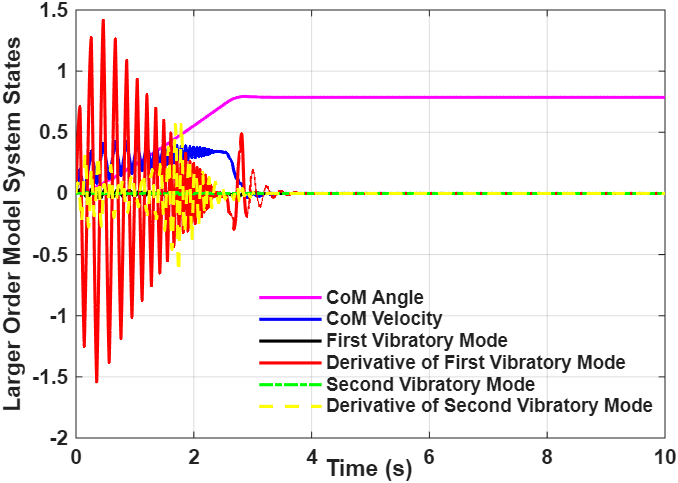}
    \caption{States Plot for Original System}
    \label{fig:fx_time_states}
\end{figure} 

\begin{figure}
    \centering
    \includegraphics[width= 8cm, height = 6cm]{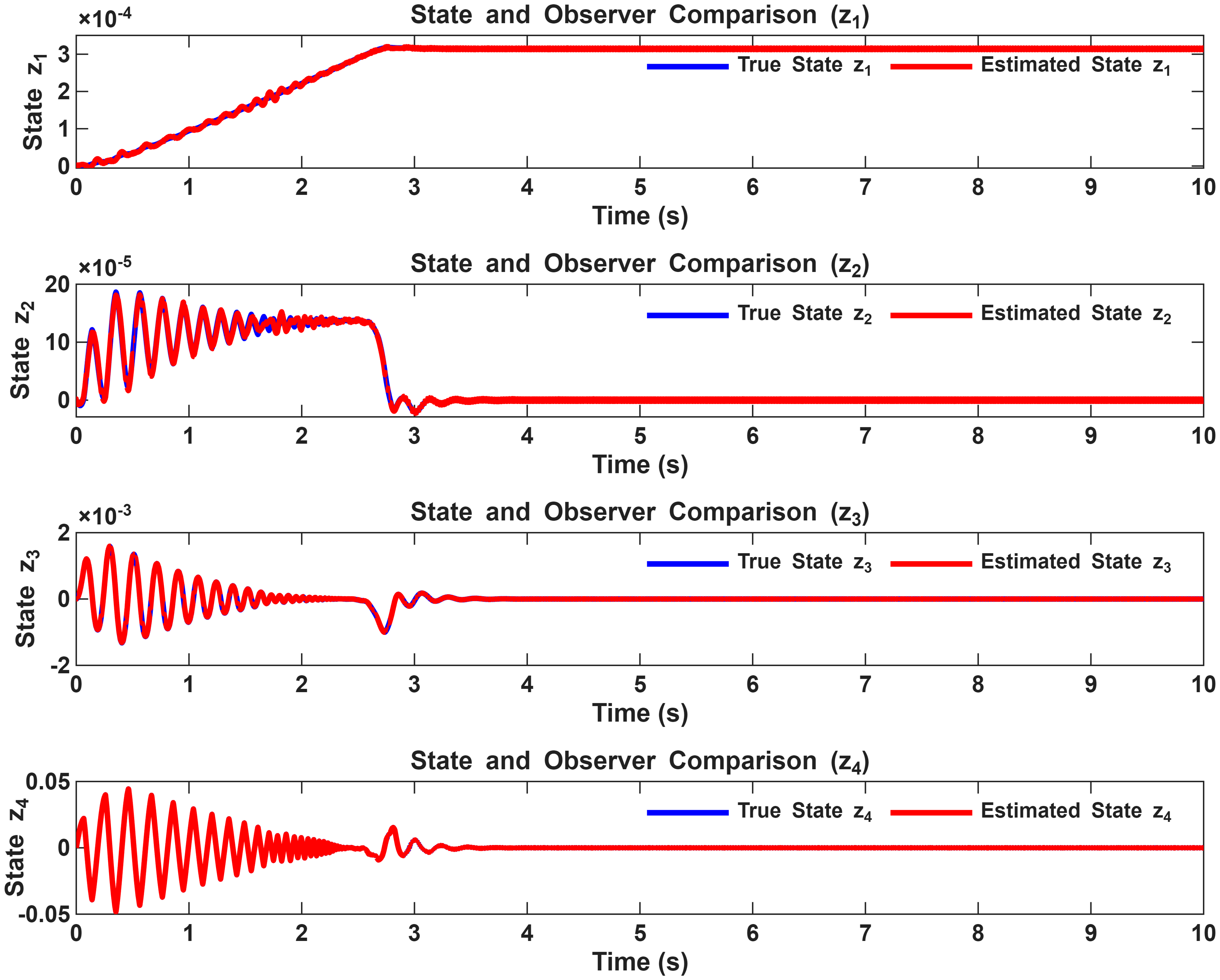}
    \caption{States and Observer Comparison}
    \label{fig:fx_time_state_obs}
\end{figure}
The performance parameters presented in Table~\ref{tab:regulation} demonstrate the effectiveness of the proposed control strategy in regulation when compared to alternative controllers. 
Table~\ref{tab:err_perf_reg} reports the integral error indices ISE, IAE, ITSE, and ITAE for both the tip angle and the joint angle. The proposed controller provides the most balanced overall behaviour, particularly in terms of reduced oscillation, faster settling, and improved joint-angle regulation. Although some tip-angle integral indices are not uniformly smaller than those of all benchmark controllers, the proposed scheme offers a superior tradeoff between transient smoothness, vibration suppression, and steady-state regulation.
\begin{table}[!ht]
\centering
\renewcommand{\arraystretch}{1.0}
\setlength{\tabcolsep}{1pt}
{
\caption{Control Performances for Regulation Problem}
\label{tab:regulation}
\begin{tabular}{|p{22mm}| p{14mm}| p{14mm}| p{17.5mm}| p{17.5mm}| }
\hline
\textbf{Control schemes} & \textbf{$T_{settling}$, $ {\theta_t(t)}$ (sec.)} & \textbf{$T_{settling}$, $ {\theta_c(t)}$ (sec.)} & \textbf{Steady state, $\|\theta_t(t)\|_2$ (rad)} & \textbf{Steady state, $\|\theta_c(t)\|_2$ (rad)}  \\
\hline
SMC 2-Mode~\cite{sharma2024position} & 3.88 & 6.4 & $3.4 \times 10^{-3}$ & $5.4 \times 10^{-3}$  \\
\hline
PD Control~\cite{zollo2007pd} & 6.1 & 5.4 & $1.03 \times 10^{-3}$ & $24.1 \times 10^{-2}$  \\
\hline
Proposed & 2.66 & 3.567 & $2.978 \times 10^{-4}$ & $2.970 \times 10^{-4}$  \\
\hline
\end{tabular}
}
\end{table}

\begin{table}[!ht]
\centering
\renewcommand{\arraystretch}{1.1}
\setlength{\tabcolsep}{1pt}
{
\caption{Comparison of error performance measures}
\label{tab:err_perf_reg}
\begin{tabular}{|c|c|c|c|c|c|c|c|c|}
\hline
 & \multicolumn{4}{c|}{\textbf{Tip Angle $\theta_t(t)$}} & \multicolumn{4}{c|}{\textbf{Joint Angle $\theta_c(t)$}} \\
\hline
\textbf{Control schemes} & \textbf{ISE} & \textbf{IAE} & \textbf{ITSE} & \textbf{ITAE} & \textbf{ISE} & \textbf{IAE} & \textbf{ITSE} & \textbf{ITAE} \\
\hline
SMC 2-Mode~\cite{sharma2024position} & 0.305 & 0.659 & 0.146 & 0.497 & 5.472 & 3.605 & 8.573 & 6.429 \\
\hline
PD Control~\cite{zollo2007pd} & 0.357 & 1.053 & 0.236 & 3.082 & 1.527 & 3.970 & 7.996 & 26.696 \\
\hline
\textbf{Proposed} & 0.651 & 1.1651 & 0.462 & 1.088 & 1.308 & 1.464 & 1.100 & 1.685 \\
\hline
\end{tabular}
}
\end{table}

In figure \ref{fig:fx_time_states}, the plot between states of the system with only the first vibratory mode considered versus time is shown. It is evident from the figure that the $CoM$ angle also converges to the desired position with suppressed vibration. The transient response is smooth, and convergence is achieved in approx. $2.3~sec$. 
\begin{figure}
    \centering
    \includegraphics[width= 8cm, height = 4cm]{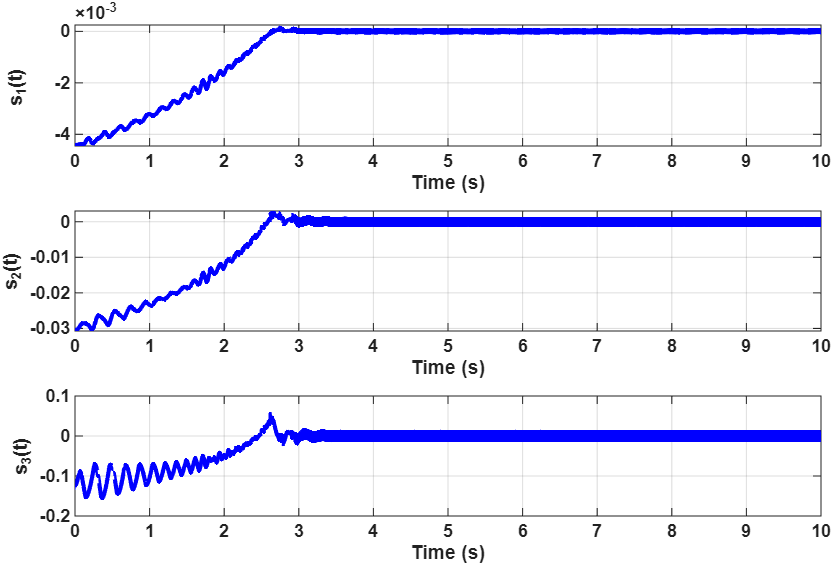}
    \caption{Sliding Variable Plots}
    \label{fig:fx_time_sliding_fn}
\end{figure}
The plot of the transformed system states versus time, along with the estimated states, is displayed in figure \ref{fig:fx_time_state_obs}. As evident in the figure, the estimated states converge rapidly to the actual system states, resulting in negligible estimation error. This demonstrates the observer's ability to accurately estimate the system states, even in the presence of disturbances.

The sliding variables plot with respect to time is given in figure \ref{fig:fx_time_sliding_fn}. The sliding variable function $s_3$ converges to the origin in a fixed time, as shown in the figure. The figure also shows that the sliding variables $s_2$ and $s_1$ converge to the origin in a fixed time.

Fig.~\ref{fig:large_order_system_states} shows time responses of the higher-order system states under the proposed control, including the CoM angle and velocity, as well as the first and second vibratory modes and their derivatives. The flexible modes and their rates exhibit bounded oscillations during the transient phase and converge rapidly to the vicinity of zero within approximately $3~sec$, indicating effective suppression of vibration dynamics. Meanwhile, the CoM angle smoothly tracks the desired profile without exciting residual flexible oscillations, demonstrating coordinated stabilisation of both rigid-body and flexible dynamics.

Fig.~\ref{fig:Sliding_var_comp}, shows sliding function response comparison under the proposed control ($s_0(t)=e$) and FOSMC ($s(t)$). The proposed scheme achieves fixed-time reaching of the sliding manifold, driving the sliding variable to a small neighbourhood of the origin within approximately $2.5~sec$ and maintaining it within a tight bound (on the order of $10^{-3}$). In contrast, FOSMC exhibits larger transient excursions and reaches a comparable neighbourhood only after about $4~sec$, with noticeable residual oscillations.
\begin{figure}
    \centering
    \includegraphics[width= 8cm, height = 4cm]{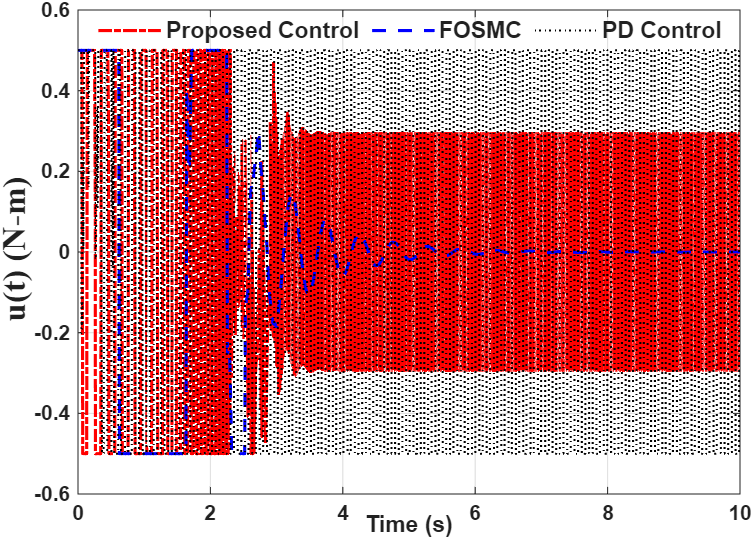}
    \caption{Larger Order System States}
    \label{fig:large_order_system_states}
\end{figure}

\begin{figure}
    \centering
    \includegraphics[width= 8cm, height = 4cm]{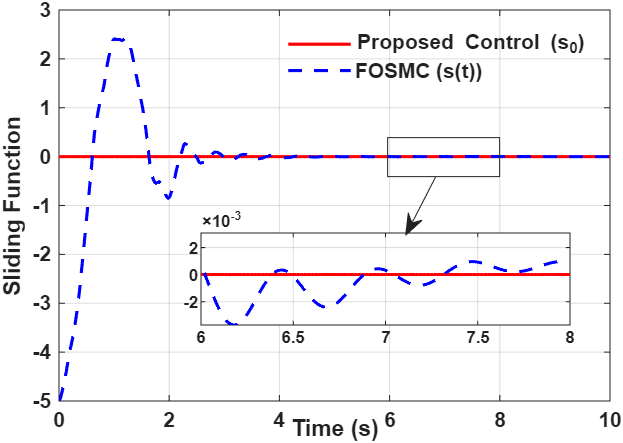}
    \caption{Sliding Variable Comparison Plot}
    \label{fig:Sliding_var_comp}
\end{figure}
\begin{figure}
    \centering
    \includegraphics[width= 8cm, height = 5cm]{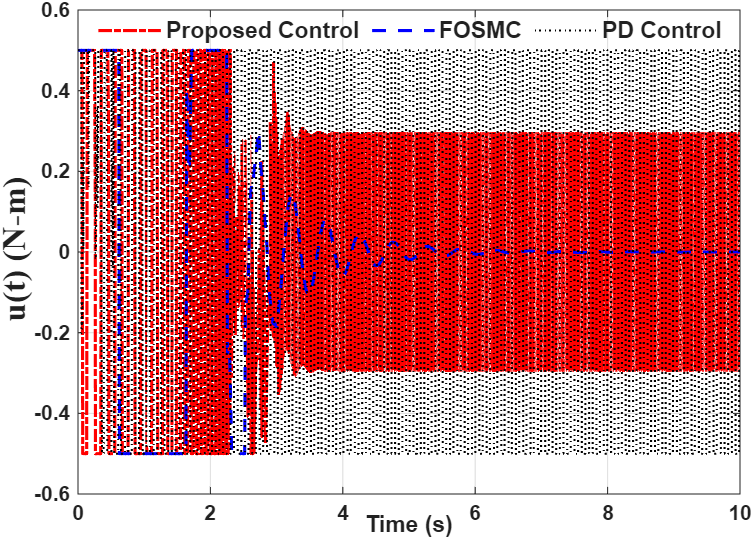}
    \caption{Control Input Plot}
    \label{fig:fx_time_control}
\end{figure}

The control input plot with respect to time is shown in Figure \ref{fig:fx_time_control}. The proposed scheme generates bounded control torques with rapidly attenuated oscillations after the transient phase ($\approx3~sec$). In contrast, FOSMC exhibits pronounced high-frequency oscillations during the reaching phase, while the PD controller produces larger steady-state control activity. The reduced oscillatory content and bounded amplitude of the proposed control indicate improved control smoothness and reduced actuator stress.

\section{Experimental Validation}
\label{sec:experimental_validation}
This section details the experimental validation of the proposed control strategy. The controller is implemented and evaluated on the Rotary Flexible Link setup developed by Quanser Inc., as described in \cite{apkarian2011workbook}. The experiments are conducted at the Control and Automation Laboratory, Department of Electrical Engineering, Indian Institute of Technology Delhi.
\begin{figure}
    \centering
    \includegraphics[width=8cm, height=5.0cm]{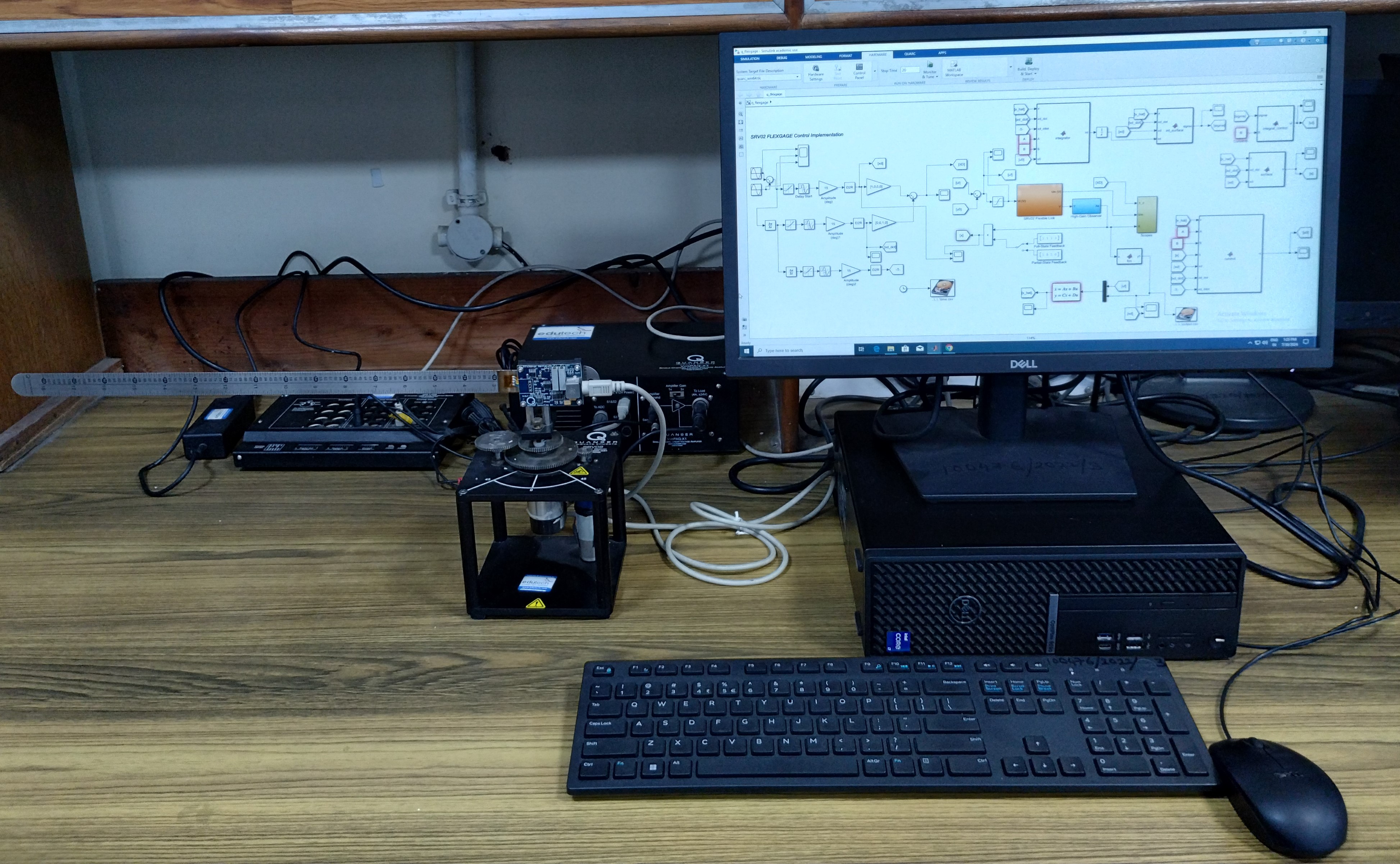}
    \caption{Complete hardware experimental set-up comprising the 
    flexible link, host system, power amplifier, and data 
    acquisition board.}
    \label{fig:experimental_setup}
\end{figure}
The experimental setup comprises a lightweight, stainless-steel, flexible link attached to a Quanser SRV02 DC servomotor. The angular position of the motor shaft is measured using a high-resolution optical encoder, while the link deflection is obtained from a strain gauge mounted near its base. The servomotor is actuated via a Quanser VoltPAQ-X1 power amplifier.

The control algorithm is implemented in real time on a host computer running MATLAB/Simulink with the QUARC real-time control software environment, interfacing with the hardware via a Q8-USB data acquisition board. The proposed controller is developed in Simulink and deployed to the experimental platform at a fixed sampling interval. During operation, the measured signals are fed back to the controller, while the unmeasured states are estimated using the designed observer.

\begin{figure}
    \centering
    \subfloat[The position of the tip]{%
        \includegraphics[width=0.9\columnwidth]{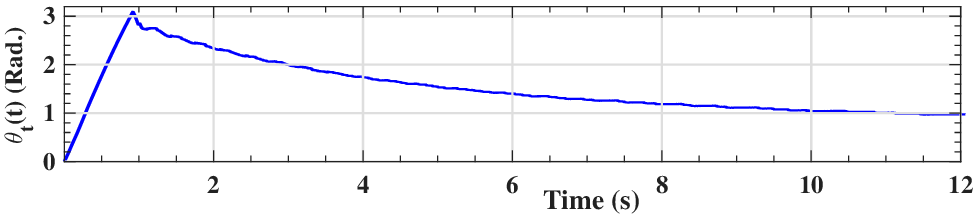}%
        \label{fig:state1}}\\
    \subfloat[Control input (motor voltage) applied to the SRV02]{%
        \includegraphics[width=0.85\columnwidth]{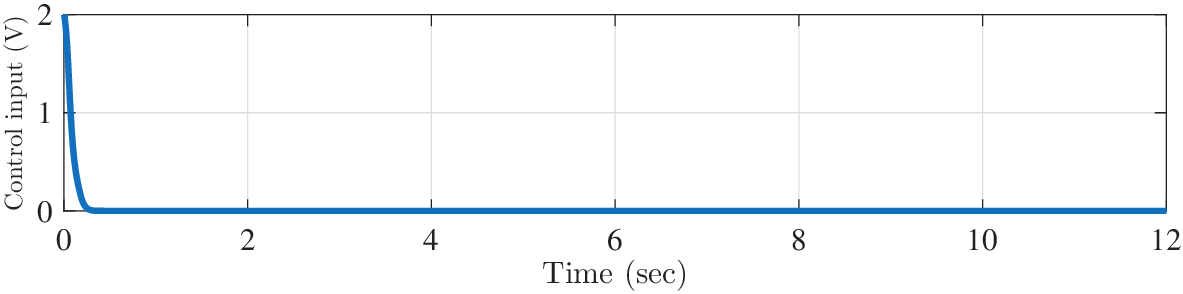}%
        \label{fig:exp_control}}
    \caption{Experimental time responses of the servo angle $x_1$ and tip deflection $x_2$ of the Quanser Rotary Flexible Link along with the proposed control law.}
    \label{fig:exp_states}
\end{figure}
The experimental results obtained from the Rotary Flexible Link set-up are presented in Fig.~\ref{fig:exp_states} and 
Fig.~\ref{fig:exp_control}. Figure~\ref{fig:exp_states} shows the time evolution of the four system states, while Fig.~\ref{fig:exp_control} depicts the corresponding control input applied to the DC motor.

The closed-loop responses obtained from the Quanser Rotary Flexible Link platform are shown in Fig.~\ref{fig:state1}, where the solid red curve denotes the measured angle with the tip deflection and the dashed blue curve denotes the corresponding estimate produced by the observer. The servo angle $x_1$ is driven from rest toward the desired set-point of $\pi/4~\mathrm{rad}$, which it reaches within roughly one second and subsequently holds with only a small steady-state offset. A brief transient excursion due to tip deflection is observed and then damped smoothly, confirming that the elastic mode of the flexible link is effectively suppressed by the proposed controller. Throughout the response, the estimated trajectories remain in close agreement with the measured ones, indicating that the observer reconstructs the internal states with good fidelity. The corresponding control effort applied to the DC motor, shown in Fig.~\ref{fig:exp_control}, stays well within the actuator saturation limits and exhibits no excessive chattering. Taken together, the hardware results are in good agreement with the theoretical predictions and demonstrate the practical applicability of the proposed control strategy on a physical flexible-link system.


\section{Conclusion}
This paper presented an observer-based composite control scheme for tip-position regulation of a single-link flexible manipulator. A nested non-singular terminal sliding-mode controller was developed to achieve practical fixed-time regulation of the tip position, while a nonlinear sliding-mode observer was introduced to reconstruct the transformed states from the measured outputs. Lyapunov-based analysis established practical fixed-time convergence of both the controller surfaces and the observer error, with explicit settling-time bounds. Numerical simulations and hardware experiments demonstrated improved transient behaviour, reduced oscillation, and smoother control effort compared with benchmark controllers. Future work will focus on extending the proposed framework to higher-mode models and on deriving exact fixed-time observer-controller synthesis under reduced-order modelling mismatch.

\bibliographystyle{unsrt}  
\bibliography{references}

\end{document}